\documentclass{article}
\usepackage{authblk}
\usepackage{subcaption}
\usepackage{pstricks,pst-tree,pst-node} 
\usepackage{graphicx,multicol}
\usepackage{epic,eepic,epsfig}
\usepackage{amssymb}
\usepackage{amsmath,amsthm}
\usepackage{color}
\usepackage{tikz}
\usepackage{epsfig,url}
\usepackage{multirow}

\usepackage{amsfonts}
\usepackage{setspace}
\usepackage{geometry}
\usetikzlibrary{shapes,snakes}

\usepackage{eso-pic} 

\newcommand{\be}{\begin{enumerate}}
\newcommand{\ee}{\end{enumerate}}
\newcommand{\bd}{\begin{description}}
\newcommand{\ed}{\end{description}}

\newcommand{\beq}{\begin{equation}}
\newcommand{\eeq}{\end{equation}}

\renewenvironment{proof}[1][]{\par \noindent {\bf Proof#1}.\ }{\hfill$\Box$
\par \vspace{11pt}}

\newtheorem{theorem}{Theorem}[section]
\newtheorem{lemma}[theorem]{Lemma}

\theoremstyle{definition}

\newcommand{\pf}{{\bf Proof: }}

\newcommand{\AY}[1]{{\color{blue}#1}}

\newcommand{\zeroTOi}[1]{[0,#1]}

\begin{document}
\bibliographystyle{plain}

\newcommand{\defparproblem}[4]{
  \vspace{3mm}
\noindent\fbox{
  \begin{minipage}{.95\textwidth}
  \begin{tabular*}{\textwidth}{@{\extracolsep{\fill}}lr} #1  & {\bf{Parameter:}} #3 \\ \end{tabular*}
  {\bf{Input:}} #2  \\
  {\bf{Question:}} #4
  \end{minipage}
  }
  \vspace{2mm}
}

\newcommand{\defproblem}[3]{
  \vspace{3mm}
\noindent\fbox{
  \begin{minipage}{.95\textwidth}
  \begin{tabular*}{\textwidth}{@{\extracolsep{\fill}}lr} #1  \\ \end{tabular*}
  {\bf{Input:}} #2  \\
  {\bf{Question:}} #3
  \end{minipage}
  }
  \vspace{2mm}
}

\newcommand{\defoptproblem}[3]{
  \vspace{3mm}
\noindent\fbox{
  \begin{minipage}{.95\textwidth}
  \begin{tabular*}{\textwidth}{@{\extracolsep{\fill}}lr} #1  \\ \end{tabular*}
  {\bf{Input:}} #2  \\
  {\bf{Find:}} #3
  \end{minipage}
  }
  \vspace{2mm}
}

\begin{spacing}{1.2}

\title{The directed 2-linkage problem with length constraints\thanks{Research supported by the Danish research council under grant number DFF-7014-00037B.}}

\author{J.Bang-Jensen\thanks{Department of Mathematics and Computer Science, University of Southern Denmark, Odense, Denmark (email: jbj@imada.sdu.dk).}, T. Bellitto\thanks{Department of Mathematics and Computer Science, University of Southern Denmark, Odense, Denmark (email: bellitto@imada.sdu.dk).}, 
W. Lochet\thanks{University of Bergen, Norway (email: william.lochet@gmail.com).}, 
A. Yeo\thanks{Department of Mathematics and Computer Science, University of Southern Denmark, Odense, Denmark (email: yeo@imada.sdu.dk).}}
\date{\today}
\date{}
\maketitle

\begin{abstract}
  The {\sc weak 2-linkage} problem for digraphs asks for a given digraph and vertices $s_1,s_2,t_1,t_2$ whether $D$ contains a pair of arc-disjoint paths $P_1,P_2$ such that $P_i$ is an $(s_i,t_i)$-path. This problem is  NP-complete for general digraphs  but polynomially solvable for acyclic digraphs \cite{fortuneTCS10}. Recently it was shown \cite{bercziESA17} that if $D$ is equipped with a weight function $w$ on the arcs
  which satisfies that all edges have positive weight, then there is a polynomial algorithm for  the variant of the weak-2-linkage problem when both paths have to be shortest paths in $D$. In this paper we consider the unit weight case and prove that for every pair constants $k_1,k_2$, there is a polynomial algorithm which decides whether the input digraph $D$ has a pair of arc-disjoint paths $P_1,P_2$ such that $P_i$ is an $(s_i,t_i)$-path and the length of $P_i$ is no more than $d(s_i,t_i)+k_i$, for $i=1,2$, where $d(s_i,t_i)$ denotes the length of the shortest $(s_i,t_i)$-path. We prove that, unless the exponential time hypothesis (ETH) fails, there is no polynomial algorithm for deciding the existence of a solution $P_1,P_2$ to the {\sc weak 2-linkage} problem where each  path $P_i$ has length at most $d(s_i,t_i)+ c\log^{1+\epsilon}{}n$ for some constant $c$.
  We also prove that the {\sc weak 2-linkage} problem remains NP-complete if we require one of the two paths to be a shortest path while the other path has no restriction on the length.\\
  
  \end{abstract}
\section{Introduction}

Notation throughout this paper follows \cite{bang2018,bang2009}. We use $\zeroTOi{i}$ to denote the set $\{0,1,2, \ldots, i\}$.


Problems concerning disjoint paths with prescribed end vertices in graphs and digraphs play an important role in many combinatorial problems. Among the most important such problems  are  
the {\sc $k$-linkage} problem and the {\sc weak $k$-linkage} problem which we formulate below for digraphs.\\

\defproblem{{\sc $k$-linkage }}{A digraph
  $D=(V,A) $ and distinct vertices $s_1,s_2,\ldots{},s_k,t_1,t_2,\ldots{},t_k$}
{Does $D$ contain $k$ vertex-disjoint paths $P_1,P_2,\ldots{},P_k$ such that $P_i$ is an $(s_i,t_i)$-path for $i\in [k]$?}\\

\defproblem{{\sc weak $k$-linkage }}{A digraph
  $D=(V,A) $ and not necessarily distinct vertices $s_1,s_2,\ldots{},s_k,t_1,t_2,\ldots{},t_k$}
{Does $D$ contain $k$ arc-disjoint paths $P_1,P_2,\ldots{},P_k$ such that $P_i$ is an $(s_i,t_i)$-path for $i\in [k]$?}\\ 

It is an easy and well-known fact that the  {\sc $k$-linkage} problem and the {\sc weak $k$-linkage} problems are polynomially equivalent in the sense that one can easily reduce one to the other by a polynomial reduction see e.g. \cite[Chapter 10]{bang2009}.  

A famous and very important result by Robertson and Seymour \cite{robertsonJCT63}  shows that the corresponding linkage problems for undirected graphs are  polynomially solvable for fixed $k$ and that the problems are  in fact FPT, meaning that there is an algorithm for each  problem whose running time is of the form $O(f(k)n^c)$ for some computable function $f$ and a constant $c$. It was shown in \cite{robertsonJCT63} that $c=3$ will do and this has been improved to $c=2$ in \cite{kawarabayashiJCT102}.

For directed graphs the situation is quite different: Fortune, Hopcroft and Wyllie \cite{fortuneTCS10} proved that already the {\sc 2-linkage} and the {\sc weak 2-linkage} problem are  NP-complete. They also showed that if the input  is an acyclic digraph, then both linkage problems are polynomially solvable when the number of terminals is fixed (not part of the input).

\begin{theorem}\cite{fortuneTCS10}
  \label{acycliclink}
  The weak {\sc $k$-linkage} problem in acyclic digrpahs is solvable in time $O(k!n^{k+2})$.
\end{theorem}

Eilam-Tzoref \cite{eilamDAM85} proved that for undirected graphs the {\sc 2-linkage} problem is also polynomially solvable if each  edge of the input graph is equipped with a positive length and the goal is to check whether there is a solution $P_1,P_2$ such that $P_i$ is a shortest $(s_i,t_i)$-path for $i=1,2$. This was recently generalized to digraphs by Berczi and Kabayashi \cite{bercziESA17} who proved the following:

\begin{theorem}
  \label{shortest2link}
  There exists  a polynomial algorithm for the following problem as well as its arc-version. Given a digraph $D=(V,A)$, vertices $s_1,s_2,t_1,t_2\in V$, a weight function $w$ on $A$ such that the weight of every directed cycle is positive and numbers $a_1$ and $a_2$; decide whether $D$ has disjoint paths $P^1_1,\dots,P^1_{a_1},P^2_1,\dots,P^2_{a_2}$ such that $P^i_j$ is a shortest  $(s_i,t_i)$-path for $i=1,2$ and $1\leqslant j\leqslant a_i$.
\end{theorem}

There are several other papers dealing with  shortest path version  of the {\sc $k$-linkage} problem, see e.g. \cite{bjorklundICALP14,schrijverTA7,kobayashiDO7}.

In this paper we consider the following  variant of the weak-2-linkage problem  where the paths do not have to be shortest paths (in terms of number of arcs)  but there is a bound on how far from being shortest they can be.

\defproblem{{\sc short weak 2-linkage } SW2L($D,s_1,s_2,t_1,t_2,k_1,k_2$)}{A digraph $D=(V,A)$, vertices $s_1,s_2,t_1,t_2\in V$ and  natural numbers $k_1,k_2$}{Is there a pair of arc-disjoint paths $P_1,P_2$ such that $P_i$ is an $(s_i,t_i)$-path and\\ $|A(P_i)|\leq d(s_i,t_i)+k_i$?}\\

Clearly this problem is NP-complete when $k_1,k_2=n-1$ since that puts no restriction on $P_1,P_2$ in a solution. The main result of our paper is that when $k_1,k_2$ are both constants the {\sc short weak 2-linkage } problem can be solved in polynomial time. We also prove that the problem is NP-complete when there is no restriction on the length of one of the paths. Finally, we show that under the exponential time hypothesis, there is no polynomial algorithm for the {\sc short weak 2-linkage } when $k_1,k_2\in O(\log^{1+\epsilon}n)$ no matter how small the value of $\epsilon$ is as long as it is positive.


\section{2-linkage with almost shortest paths}

Let $D$ be a digraph and $s$ a vertex. The {\bf reach} of $s$ is the set of vertices $x$ such that there exists
a path from $s$ to $x$ in $D$. Using breath-first-search we can partition the reach of a vertex $s$ into {\bf levels}, 
such that $L_s^i$ denotes the set of vertices $x$ such that the shortest path from $s$ to $x$ is of length $i$. We say that an arc $uv$ is {\bf between} two levels if, $d(s,u)<\infty$ and  $d(s,v)=d(s,u)+1$.

Suppose $s_1$ and $s_2$ are fixed, let $A_1$ denote the set of arcs between two consecutive levels  from $s_1$ 
and $A_2$ the set of arcs between two consecutive levels from $s_2$. Note that both $A_1$ and $A_2$ form acylic digraphs. Furthermore, an arc $uv$ is in $A_i$ if and only if some shortest $(s_i,v)$-path uses the arc $uv$. 
We will use the following lemma:

\begin{lemma}\label{lemma:paths}
    If $P$ is a path from $s_1$ to $t_1$ of length at most $d(s_1, t_1) + k$, then $P$ uses at most $k$ arcs not belonging to $A_1$. 
\end{lemma}

\begin{proof}
Every path from $s_1$ to $t_1$ must visit every level with index smaller than $d(s_1, t_1)$ at least once. 
Moreover, it must use an arc of $A_1$ to go from one level to the next, which ends the proof.  
\end{proof}

\begin{theorem}
  \label{thm:SWLalg}
  For every fixed choice of positive integers $k_1,k_2$ the problem {\sc short weak linkage problem} with input $[D,s_1,s_2,t_1,t_2,k_1,k_2]$ is polynomially solvable. 
\end{theorem}

\pf
Let $k=\max\{k_1,k_2\}$.
We shall describe an algorithm that runs in $n^{O(k)}$ for the problem. 
Let $E_1 = (v_1, u_1), \dots, (v_i, u_i)$ and $E_2 = (z_1, w_1), \dots, (z_j, w_j)$ be two ordered sets of at most $k$ arcs each.
Recall that  $A_1$ denotes the set of arcs between two consecutive levels from $s_1$
and $A_2$ the set of arcs between two consecutive levels from $s_2$. 
Let $d_{A_{\ell}}(x,y)$ denote the distance from $x$ to $y$ in the digraph induced by the arcs in $A_{\ell}$.
We call $E_1$ and $E_2$ {\em feasible} if the following holds.

\begin{itemize}
\item $d_{A_1}(s_1,v_1) + 1 + d_{A_1}(u_1,v_2) + 1 + d_{A_1}(u_2,v_3)  + \cdots + d_{A_1}(u_{i-1},v_i) + 1 + d_{A_1}(u_i,t_1) \leqslant d_D(s_1,t_1)+k_1$
\item $d_{A_2}(s_2,z_1) + 1 + d_{A_2}(w_1,z_2) + 1 + d_{A_2}(w_2,z_3)  + \cdots + d_{A_2}(w_{j-1},z_i) + 1 + d_{A_2}(w_j,t_2) \leqslant d_D(s_2,t_2)+k_2$
\end{itemize}


We will describe a $O(n^C)$  algorithm  for some constant $C$, which  decides if there exists a solution $P_1$, $P_2$ to the problem
such that for $\ell=1,2$, $P_{\ell}$ only uses arcs of $A_{\ell}$ and $E_{\ell}$. 
To solve the general question, we only need to run this algorithm for all feasible choices of $E_1$
and $E_2$. As there are at most ${m \choose k}$ ways of choosing a $k$-feasible set, we note that there are less than $(m^k)^2 \leq n^{4k}$
(as $m \leq n^2$) ways of choosing $E_1$ and $E_2$. So the algorithm only needs to be run at most $O(n^{4k})$ times.

Let now $E_1$ and $E_2$ be fixed. We create the digraph $D'$ by adding the vertices $s_1',t_1',s_2',t_2'$ to $D$ and the following paths:

\begin{itemize}
\item A path from $s_1'$ to every vertex $x \in \{s_1,u_1,u_2,\ldots,u_i\}$ of length $d_D(s_1,x)+1$. 
\item A path from every vertex $x \in \{t_1,v_1,v_2,\ldots,v_i\}$ to $t_1'$ of length $d_D(x,t_1)+1$. 
\item A path from $s_2'$ to every vertex $x \in \{s_2,w_1,w_2,\ldots,w_j\}$ of length $d_D(s_2,x)+1$. 
\item A path from every vertex $x \in \{t_2,z_1,z_2,\ldots,z_j\}$ to $t_2'$ of length $d_D(x,t_2)+1$.      
\end{itemize}

Note that the internal vertices on all the above added paths are distinct and new vertices. 
Let $P_1^*$ be a $(s_1',t_1')$-path in $D'$ and let $x$ be the first vertex on $P_1^*$ from $\{s_1,u_1,u_2,\ldots,u_i\}$
and let $y$ be the last vertex on $P_1^*$ from $\{t_1,v_1,v_2,\ldots,v_i\}$. Then the subpath from $s_1'$ to $x$ has length 
$d_D(s_1,x)+1$ and the subpath from $y$ to $t_1'$ has length $d_D(y,t_1)+1$. This implies that the length of $P_1^*$ is the 
following.
\[
(d_D(s_1,x)+1)+(d_D(y,t_1)+1)+d_D(x,y) = d_D(s_1,x)+d_D(x,y)+d_D(y,t_1)+2 \geq d_D(s_1,t_1)+2
\]

As there exists a $(s_1',t_1')$-path of length $d_D(s_1,t_1)+2$ in $D'$ (using the arcs $s_1' s_1$ and $t_1 t_1'$ and a shortest
$(s_1,t_1)$-path in $D$), we note that the shortest 
$(s_1',t_1')$-path in $D'$ has length exactly $d_D(s_1,t_1)+2$. Furthermore if the subpath of $P_1^*$ from $x$ to $y$ only uses arcs from
$A_1$ then it has length $d_D(x,y)$ and we have equality everywhere in the above equation, which implies that the length of $P_1'$ is 
$d_D(s_1,t_1)+2 = d_{D'}(s_1',t_1')$. Analogously if the length of $P_1^*$ is $d_{D'}(s_1',t_1')$ then the subpath from $x$ to $y$ only uses
arcs from $A_1$.

Clearly the analogous result also holds for a shortest path from $s_2'$ to $t_2'$.
By Theorem \ref{shortest2link}, we know that we can determine in polynomial time if there exist $i+1$ paths from $s_1$ to $t_1$ and $j+1$ paths from $s_2$ to $t_2$ such that all $i+j+2$ paths are arc-disjoint.


We claim that if such paths exist, then the answer to our instance of the {\sc short weak linkage problem}
is {\em true} and if there is no such $i+j+2$ arc-disjoint paths for any feasible choice of $E_1$ and $E_2$, then the answer to
our instance is {\em false}.

First assume that we found $i+j+2$ arc-disjoint paths for some feasible choice of $E_1$ and $E_2$. Now remove all vertices in $V(D') \setminus V(D)$
from the $(s_1',t_1')$-paths and add the arcs $E_1$. Note that the outdegree of $s_1$ will be one (as it belongs to one of the paths) and the indegree
of $s_1$ will be zero. Analogously the indegree of $t_1$ will be one and the outdegree will be zero. All other vertices will have indegree and out degree equal to each other (if they belong to $k$ paths then the indegree and outdegree will both be $k$).
Therefore the arcs in the resulting subdigraph form a path from $s_1$ to $t_1$ plus possibly a number of cycles.  
As the total number of arcs in the subdigraph is less than
$d(s_1,t_1)+k_1$ the path from $s_1$ to $t_1$ (after discarding any cycles) also has length less than $d(s_1,t_1)+k_1$. Indeed, the graph contains all the arcs of $E_1$ ($i$ arcs) and all the arcs of the $i+1$ paths of length $d(s_1,t_1)+2$ except those of $D'\setminus D$. The number of arcs of $D'\setminus D$ in those paths is $N_1=\sum_{x\in\{s_1,u_1,\dots,u_i\}}(d(s_1,x)+1)+\sum_{x\in\{v_1,\dots,v_i,t_1\}}(d(x,t_1)+1)$. 
Let us set $u_0=s_1$ and $v_{i+1}=t_1$. The total number of arcs in our graph is:
\[\begin{split}
i+(d(s_1,t_1)+2)\times (i+1)-N_1&=i+\sum_{\ell=0}^i d(s_1,t_1)-d(s_1,u_{\ell})-d(v_{\ell+1},t_1)\\
&\leqslant i+\sum_{\ell=0}^i d(u_{\ell},v_{\ell+1})\\
&\leqslant d(s_1,t_1)+k_1\quad\text{by definition of the feasibility of }E_1
  \end{split}
\]

Analogously we find a path from $s_2$ to $t_2$ of length less than $d(s_2,t_2)+k_2$. By our construction these paths are arc-disjoint, completing the proof of one direction.

Now assume that there exist arc-disjoint paths $P_1$ and $P_2$ in $D$, such that $P_{\ell}$ is a $(s_{\ell},t_{\ell})$-path of length less than $d(s_{\ell},t_{\ell})+k_{\ell}$.
Let $E_{\ell}$ be the arcs on $P_{\ell}$ that do not belong to $A_{\ell}$ (${\ell} \in [2]$). Hence, $E_1$ and $E_2$ are feasible. Using these $E_1$ and $E_2$ and the 
subpaths of $P_1$ and $P_2$ after removing the arcs in $E_1$ and $E_2$ we note that we can obtain the desired $i+j+2$ arc-disjoint paths in $D'$.
This completes the proof.

\vspace{0.4cm}

$D'$ has size $|A(D)|^2$, so the existence of this path $Q$ can be checked in polynomial time, 
and the overall problem can be solved in time $n^{O(k)}$.

\section{Non-polynomial cases}

This subsection is devoted to the proof of the NP-completeness of the problem of semi-short weak 2-linkage:

\defproblem{{\sc semi-short weak 2-linkage } SSW2L($D,s_1,s_2,t_1,t_2,k$)}{A digraph $D=(V,A)$, vertices $s_1,s_2,t_1,t_2\in V$ and a natural number $k$}{Is there a pair of arc-disjoint paths $P_1,P_2$ such that $P_1$ is an $(s_1,t_1)$-path of length \\ $|A(P_1)|\leq d(s_1,t_1)+k$ and $P_2$ is an $(s_2,t_2)$-path?}


\begin{theorem}
  \label{oneshort}
  The semi-short weak 2-linkage problem  is NP-complete for all values of $k$.
\end{theorem}

\pf Let $\cal F$ be an instance of 3-SAT with $m$ clauses  all of which consists of three literals over  the variables $x_1,x_2,\ldots{},x_n$. We may assume that every variable appears as a literal in at least one clause.  
For each $i\in [n]$ let $a_i$ be the maximum number of times that $x_i$ occurs as the same  literal ($x_i$ or $\bar{x}_i$).
Fix an  ordering $C_1,C_2,\ldots{},C_m$ of the clauses. This induces an ordering of the occurences of each literal.\\

We now define a digraph $D=D({\cal F})$ and vertices $s_1,s_2,t_1,t_2$ such that $D$ contains arc-disjoint paths $P_1,P_2$ satisfying  that $P_1$ is a $(s_1,t_1)$-path of length at most $d(s_1,t_1)+k$ and $P_2$ is an $(s_2,t_2)$-path in $D$ if and only if $\cal F$ is satisfiable. We start by defining a gadget $W_i$ corresponding to the variable $x_i$ in $\cal F$ for each $i\in [n]$. The digraph $W_i$ consists of two directed $(u_i,v_i)$-paths,
$T_i=u_iy_{(i,2a_i)}y_{(i,2a_i-1)}\ldots{}y_{(i,2)}y_{(i,1)}v_i$ and $F_i=u_i\bar{y}_{(i,2a_i)}\bar{y}_{(i,2a_i-1)}\ldots{}\bar{y}_{(i,2)}\bar{y}_{(i,1)}v_i$. We first build an intermediate digraph $D'$ and then modify this to get $D$. The vertices of $D'$ will consist of all the vertices of $W_1,W_2,\ldots{},W_n$, a set of vertices $\{c_1,c'_1,\ldots{},c_m,c'_m\}$ and finally 4 new vertices $s_1,s_2,t_1,t_2$. We link these together as follows:
\begin{itemize}
\item Add an arc from $v_i$ to $u_{i+1}$ for $i\in [n-1]$.
\item Add the arcs $s_1u_1,v_nt_1$.
\item Add the arcs $s_2c_1,c'_mt_2$
\item Add the arcs $c'_ic_{i+1}$ for $i\in [m-1]$.
  \item Finally we add the arcs that mimic the clauses of $\cal F$. For each clause we add three arcs going out of $c_j$ and three arcs entering $c'_j$ as follows: let $C_j=(\ell_{j,1}\vee \ell_{j,2}\vee\ell_{j,3})$. If $\ell_{j,i}$ is the variable $x_i$ and this is the $r$'th occurence of $x_i$ according to the ordering of the clauses, then we add the arcs $c_jy_{(i,2r)},y_{(i,2r-1)}c'_j$. Similarly, if 
    $\ell_{j,i}$ is the variable $\bar{x}_f$ and this is the $c$'th occurrence of $\bar{x}$ according to the ordering of the clauses, then we add the arcs
    $c_j\bar{y}_{(f,2c)},\bar{y}_{(f,2c-1)}c'_j$. We do the same for the remaining two literals of $C_j$.
  \end{itemize}

 To obtain $D$ from $D'$ we subdivide each of the arcs we described in the last bullet above $mn+2n+k$ times. This last step is there to make sure 
that any $(s_1,t_1)$-path in $D$ of length at most $d(s_1,t_1)+k$ uses no vertex in $\{c_1,c'_1,\ldots{},c_m,c'_m\}$ 
(subdividing $(2n+k+\sum_{i=1}^n a_i)/2$ times would also suffice). 
Note that every $(s_1,t_1)$-path which avoids $\{c_1,c'_1,\ldots{},c_m,c'_m\}$ is a shortest $(s_1,t_1)$-path.

  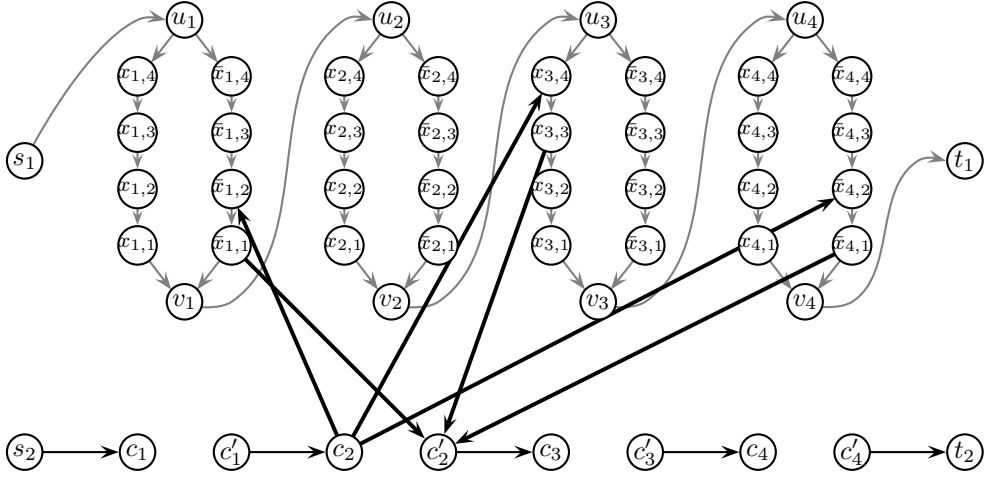
\begin{figure}[!h]
\begin{pspicture}(-0.2,-0.2)(13.2,6.45)

\pscurve[linecolor=gray]{->, arrowsize=5pt}(0.25,4.125)(1.75,6)(2.125,6)

\psline[linecolor=gray]{->, arrowsize=5pt}(2.375,6)(2.82,5.43)
\psline[linecolor=gray]{->, arrowsize=5pt}(2.375,6)(1.93,5.43)
\psline[linecolor=gray]{->, arrowsize=4pt}(3,5.25)(3,4.77)
\psline[linecolor=gray]{->, arrowsize=4pt}(3,4.5)(3,4.02)
\psline[linecolor=gray]{->, arrowsize=4pt}(3,3.75)(3,3.27)
\psline[linecolor=gray]{->, arrowsize=4pt}(1.75,5.25)(1.75,4.77)
\psline[linecolor=gray]{->, arrowsize=4pt}(1.75,4.5)(1.75,4.02)
\psline[linecolor=gray]{->, arrowsize=4pt}(1.75,3.75)(1.75,3.27)
\psline[linecolor=gray]{->, arrowsize=5pt}(3,3)(2.54,2.41)
\psline[linecolor=gray]{->, arrowsize=5pt}(1.75,3)(2.2,2.41)

\pscurve[linecolor=gray]{->, arrowsize=5pt}(2.125,2.25)(3,2.25)(4.5,6)(4.875,6)

\psline[linecolor=gray]{->, arrowsize=5pt}(5.125,6)(5.57,5.43)
\psline[linecolor=gray]{->, arrowsize=5pt}(5.125,6)(4.68,5.43)
\psline[linecolor=gray]{->, arrowsize=4pt}(5.75,5.25)(5.75,4.77)
\psline[linecolor=gray]{->, arrowsize=4pt}(5.75,4.5)(5.75,4.02)
\psline[linecolor=gray]{->, arrowsize=4pt}(5.75,3.75)(5.75,3.27)
\psline[linecolor=gray]{->, arrowsize=4pt}(4.5,5.25)(4.5,4.77)
\psline[linecolor=gray]{->, arrowsize=4pt}(4.5,4.5)(4.5,4.02)
\psline[linecolor=gray]{->, arrowsize=4pt}(4.5,3.75)(4.5,3.27)
\psline[linecolor=gray]{->, arrowsize=5pt}(5.75,3)(5.29,2.41)
\psline[linecolor=gray]{->, arrowsize=5pt}(4.5,3)(4.95,2.41)

\pscurve[linecolor=gray]{->, arrowsize=5pt}(4.875,2.25)(5.75,2.25)(7.25,6)(7.625,6)

\psline[linecolor=gray]{->, arrowsize=5pt}(7.875,6)(8.32,5.43)
\psline[linecolor=gray]{->, arrowsize=5pt}(7.875,6)(7.43,5.43)
\psline[linecolor=gray]{->, arrowsize=4pt}(8.5,5.25)(8.5,4.77)
\psline[linecolor=gray]{->, arrowsize=4pt}(8.5,4.5)(8.5,4.02)
\psline[linecolor=gray]{->, arrowsize=4pt}(8.5,3.75)(8.5,3.27)
\psline[linecolor=gray]{->, arrowsize=4pt}(7.25,5.25)(7.25,4.77)
\psline[linecolor=gray]{->, arrowsize=4pt}(7.25,4.5)(7.25,4.02)
\psline[linecolor=gray]{->, arrowsize=4pt}(7.25,3.75)(7.25,3.27)
\psline[linecolor=gray]{->, arrowsize=5pt}(8.5,3)(8.04,2.41)
\psline[linecolor=gray]{->, arrowsize=5pt}(7.25,3)(7.7,2.41)

\pscurve[linecolor=gray]{->, arrowsize=5pt}(7.625,2.25)(8.5,2.25)(10,6)(10.375,6)

\psline[linecolor=gray]{->, arrowsize=5pt}(10.625,6)(11.07,5.43)
\psline[linecolor=gray]{->, arrowsize=5pt}(10.625,6)(10.18,5.43)
\psline[linecolor=gray]{->, arrowsize=4pt}(11.25,5.25)(11.25,4.77)
\psline[linecolor=gray]{->, arrowsize=4pt}(11.25,4.5)(11.25,4.02)
\psline[linecolor=gray]{->, arrowsize=4pt}(11.25,3.75)(11.25,3.27)
\psline[linecolor=gray]{->, arrowsize=4pt}(10,5.25)(10,4.77)
\psline[linecolor=gray]{->, arrowsize=4pt}(10,4.5)(10,4.02)
\psline[linecolor=gray]{->, arrowsize=4pt}(10,3.75)(10,3.27)
\psline[linecolor=gray]{->, arrowsize=5pt}(11.25,3)(10.79,2.41)
\psline[linecolor=gray]{->, arrowsize=5pt}(10,3)(10.45,2.41)

\pscurve[linecolor=gray]{->, arrowsize=5pt}(10.375,2.25)(11.25,2.25)(12.125,4.125)(12.5,4.125)

\psline{->, arrowsize=5pt}(0.25,0.25)(1.5,0.25)
\psline{->, arrowsize=5pt}(3,0.25)(4.25,0.25)
\psline{->, arrowsize=5pt}(5.75,0.25)(7,0.25)
\psline{->, arrowsize=5pt}(8.5,0.25)(9.75,0.25)
\psline{->, arrowsize=5pt}(11.25,0.25)(12.5,0.25)


\psline[linewidth=1.5pt]{->, arrowsize=5pt}(4.5,0.25)(3.08,3.51)
\psline[linewidth=1.5pt]{->, arrowsize=5pt}(3,3)(5.57,0.43)
\psline[linewidth=1.5pt]{->, arrowsize=5pt}(4.5,0.25)(7.13,5.03)
\psline[linewidth=1.5pt]{->, arrowsize=5pt}(7.25,4.5)(5.83,0.49)
\psline[linewidth=1.5pt]{->, arrowsize=5pt}(4.5,0.25)(11.03,3.63)
\psline[linewidth=1.5pt]{->, arrowsize=5pt}(11.25,3)(5.97,0.36)

%

 \pscircle[fillstyle=solid, fillcolor=white](0.25,4.125){0.25} 

 \pscircle[fillstyle=solid, fillcolor=white](2.375,2.25){0.25}
 \pscircle[fillstyle=solid, fillcolor=white](1.75,3){0.27}
 \pscircle[fillstyle=solid, fillcolor=white](1.75,3.75){0.27}
 \pscircle[fillstyle=solid, fillcolor=white](1.75,4.5){0.27}
 \pscircle[fillstyle=solid, fillcolor=white](1.75,5.25){0.27}
 \pscircle[fillstyle=solid, fillcolor=white](3,3){0.27}
 \pscircle[fillstyle=solid, fillcolor=white](3,3.75){0.27}
 \pscircle[fillstyle=solid, fillcolor=white](3,4.5){0.27}
 \pscircle[fillstyle=solid, fillcolor=white](3,5.25){0.27}
 \pscircle[fillstyle=solid, fillcolor=white](2.375,6){0.25}
 
 \pscircle[fillstyle=solid, fillcolor=white](5.125,2.25){0.25}
 \pscircle[fillstyle=solid, fillcolor=white](4.5,3){0.27}
 \pscircle[fillstyle=solid, fillcolor=white](4.5,3.75){0.27}
 \pscircle[fillstyle=solid, fillcolor=white](4.5,4.5){0.27}
 \pscircle[fillstyle=solid, fillcolor=white](4.5,5.25){0.27}
 \pscircle[fillstyle=solid, fillcolor=white](5.75,3){0.27}
 \pscircle[fillstyle=solid, fillcolor=white](5.75,3.75){0.27}
 \pscircle[fillstyle=solid, fillcolor=white](5.75,4.5){0.27}
 \pscircle[fillstyle=solid, fillcolor=white](5.75,5.25){0.27}
 \pscircle[fillstyle=solid, fillcolor=white](5.125,6){0.25}

 \pscircle[fillstyle=solid, fillcolor=white](7.875,2.25){0.25}
 \pscircle[fillstyle=solid, fillcolor=white](7.25,3){0.27}
 \pscircle[fillstyle=solid, fillcolor=white](7.25,3.75){0.27}
 \pscircle[fillstyle=solid, fillcolor=white](7.25,4.5){0.27}
 \pscircle[fillstyle=solid, fillcolor=white](7.25,5.25){0.27}
 \pscircle[fillstyle=solid, fillcolor=white](8.5,3){0.27}
 \pscircle[fillstyle=solid, fillcolor=white](8.5,3.75){0.27}
 \pscircle[fillstyle=solid, fillcolor=white](8.5,4.5){0.27}
 \pscircle[fillstyle=solid, fillcolor=white](8.5,5.25){0.27}
 \pscircle[fillstyle=solid, fillcolor=white](7.875,6){0.25}
 
 \pscircle[fillstyle=solid, fillcolor=white](10.625,2.25){0.25}
 \pscircle[fillstyle=solid, fillcolor=white](10,3){0.27}
 \pscircle[fillstyle=solid, fillcolor=white](10,3.75){0.27}
 \pscircle[fillstyle=solid, fillcolor=white](10,4.5){0.27}
 \pscircle[fillstyle=solid, fillcolor=white](10,5.25){0.27}
 \pscircle[fillstyle=solid, fillcolor=white](11.25,3){0.27}
 \pscircle[fillstyle=solid, fillcolor=white](11.25,3.75){0.27}
 \pscircle[fillstyle=solid, fillcolor=white](11.25,4.5){0.27}
 \pscircle[fillstyle=solid, fillcolor=white](11.25,5.25){0.27}
 \pscircle[fillstyle=solid, fillcolor=white](10.625,6){0.25}

 \pscircle[fillstyle=solid, fillcolor=white](12.75,4.125){0.25}
 
 \pscircle[fillstyle=solid, fillcolor=white](0.25,0.25){0.25}
 \pscircle[fillstyle=solid, fillcolor=white](1.75,0.25){0.25}
 \pscircle[fillstyle=solid, fillcolor=white](3,0.25){0.25}
 \pscircle[fillstyle=solid, fillcolor=white](4.5,0.25){0.25}
 \pscircle[fillstyle=solid, fillcolor=white](5.75,0.25){0.25}
 \pscircle[fillstyle=solid, fillcolor=white](7.25,0.25){0.25}
 \pscircle[fillstyle=solid, fillcolor=white](8.5,0.25){0.25}
 \pscircle[fillstyle=solid, fillcolor=white](10,0.25){0.25}
 \pscircle[fillstyle=solid, fillcolor=white](11.25,0.25){0.25}
 \pscircle[fillstyle=solid, fillcolor=white](12.75,0.25){0.25}

 \rput(0.25,4.125){$s_1$}
 
 \rput(2.375,2.25){$v_1$}
 \rput(1.75,3){\footnotesize{$x_{1,1}$}}
 \rput(1.75,3.75){\footnotesize{$x_{1,2}$}}
 \rput(1.75,4.5){\footnotesize{$x_{1,3}$}}
 \rput(1.75,5.25){\footnotesize{$x_{1,4}$}}
 \rput(3,3){\footnotesize{$\bar x_{1,1}$}}
 \rput(3,3.75){\footnotesize{$\bar x_{1,2}$}}
 \rput(3,4.5){\footnotesize{$\bar x_{1,3}$}}
 \rput(3,5.25){\footnotesize{$\bar x_{1,4}$}}
 \rput(2.375,6){$u_1$}

  \rput(5.125,2.25){$v_2$}
 \rput(4.5,3){\footnotesize{$x_{2,1}$}}
 \rput(4.5,3.75){\footnotesize{$x_{2,2}$}}
 \rput(4.5,4.5){\footnotesize{$x_{2,3}$}}
 \rput(4.5,5.25){\footnotesize{$x_{2,4}$}}
 \rput(5.75,3){\footnotesize{$\bar x_{2,1}$}}
 \rput(5.75,3.75){\footnotesize{$\bar x_{2,2}$}}
 \rput(5.75,4.5){\footnotesize{$\bar x_{2,3}$}}
 \rput(5.75,5.25){\footnotesize{$\bar x_{2,4}$}}
 \rput(5.125,6){$u_2$}

  \rput(7.875,2.25){$v_3$}
 \rput(7.25,3){\footnotesize{$x_{3,1}$}}
 \rput(7.25,3.75){\footnotesize{$x_{3,2}$}}
 \rput(7.25,4.5){\footnotesize{$x_{3,3}$}}
 \rput(7.25,5.25){\footnotesize{$x_{3,4}$}}
 \rput(8.5,3){\footnotesize{$\bar x_{3,1}$}}
 \rput(8.5,3.75){\footnotesize{$\bar x_{3,2}$}}
 \rput(8.5,4.5){\footnotesize{$\bar x_{3,3}$}}
 \rput(8.5,5.25){\footnotesize{$\bar x_{3,4}$}}
 \rput(7.875,6){$u_3$}

  \rput(10.625,2.25){$v_4$}
 \rput(10,3){\footnotesize{$x_{4,1}$}}
 \rput(10,3.75){\footnotesize{$x_{4,2}$}}
 \rput(10,4.5){\footnotesize{$x_{4,3}$}}
 \rput(10,5.25){\footnotesize{$x_{4,4}$}}
 \rput(11.25,3){\footnotesize{$\bar x_{4,1}$}}
 \rput(11.25,3.75){\footnotesize{$\bar x_{4,2}$}}
 \rput(11.25,4.5){\footnotesize{$\bar x_{4,3}$}}
 \rput(11.25,5.25){\footnotesize{$\bar x_{4,4}$}}
 \rput(10.625,6){$u_4$}

  \rput(12.75,4.125){$t_1$}

 \rput(0.25,0.25){$s_2$}
 \rput(1.75,0.25){$c_1$}
 \rput(3,0.25){$c'_1$}
 \rput(4.5,0.25){$c_2$}
 \rput(5.75,0.25){$c'_2$}
 \rput(7.25,0.25){$c_3$}
 \rput(8.5,0.25){$c'_3$}
 \rput(10,0.25){$c_4$}
 \rput(11.25,0.25){$c'_4$}
 \rput(12.75,0.25){$t_2$}

\end{pspicture}
    \caption{An intermediate stage in the construction of the graph associated to the formula ${\cal F}=(x_2\lor x_3\lor x_4)\land(\bar x_1\lor x_3\lor \bar x_4)\land(x_1\lor \bar x_2\lor \bar x_3)\land(x_1\lor x_2\lor \bar x_4)$. The bold edges denote paths of length $1+mn+2n+k$. For readability, this figure does not depict the arcs associated to the clauses other than the second one. The second clause contains the first occurrence of the literal $\bar{x_1}$ and $\bar{x_4}$ in $\mathcal F$ and the second occurrence of $x_3$.}\label{NPCfig}
    \end{figure}

    Suppose first that $D$ contains arc-disjoint paths $P_1,P_2$, such that $P_1$ is a $(s_1,t_1)$-path of length at most $d(s_1,t_1)+k$ and $P_2$ is an $(s_2,t_2)$-path. Recall that $P_1$ does not use any vertices from $\{c_1,c'_1,\ldots{},c_m,c'_m\}$.
Define a truth assignment $\phi:\{x_1,\ldots{},x_n\}\rightarrow \{0,1\}$ 
(where 1 corresponds to true and 0 to false) as follows: if $P_1$ uses the path 
    $T_i$, then let $\phi{}(x_i)=0$ and if it uses the path 
    $F_i$ we let $\phi{}(x_i)=1$. 
Note that it follows from the way that we added the arcs between $\{c_1,c'_1,\ldots{},c_m,c'_m\}$ and the $W_i$'s that 
$P_2$ visits all the vertices in $\{c_1,c'_1,\ldots{},c_m,c'_m\}$ and does so in the order $c_1,c'_1,\ldots{},c_m,c'_m$
(as if $P_2$ has visited the vertices $s_2,c_1,c'_1,\ldots{},c_j',c_j$ and then for example has a path into $x_{i,r}$ in $W_i$, then 
all vertices $x_{i,r}, x_{i,r-1}, x_{i,r-2},\ldots,x_{i,1}$ only have arcs in $D'$ into $\{c_1,c'_1,\ldots{},c_j,c'_{j+1}\}$).
Furthermore,  as  $P_2$ is arc-disjoint from $P_1$, it must hold that for each $j\in [m]$
at least one of the three $W_i$'s to which $c_j$ has an arc this arc goes to the opposite path to the one that $P_1$ used and 
hence $C_j$ will be true under $\phi{}$.\\

Assume now that $\phi$ is a satisfying truth assignment for $\cal F$. Then we construct the path $P_1$ as follows: if $\phi{}(x_i)=1$, then we use the path
$F_i$ inside $W_i$ and otherwise we use the path $T_i$. Finally we add the arcs between the $W_i$'s as well as $s_1u_1,v_nt_1$. By the construction of $D$, this $P_1$ is a shortest $(s_1,t_1)$-path. Now we show how to construct $P_2$ so that it is arc-disjoint from $P_1$. As $\phi$ satisfies each clause $C_j$ we can fix one literal $\ell$ of that is true under $\phi$ and then use the arcs between $c_j,c'_j$ and  the two vertices of the corresponding path inside that $W_i$ for which $\ell$ is a literal over $x_i$. It follows from the way we routed $P_1$ that $P_1$ does not use any arc on that path. Hence doing this for all the clause vertices and finally adding the arcs $s_2c_1,c'_mt_2$ and $c'_jc_{j+1}$ for $j\in [m-1]$ we obtain the desired path $P_2$.
\qed

\vspace{2mm}

Note that our algorithm in Theorem \ref{thm:SWLalg} is only polynomial for constant $k$, so it is natural to ask whether we could replace constant $k$ by some function of $n$.

Recall the so-called Exponential Time Hypothesis (ETH) which in one of many formulations says that there exist a real number $\delta>0$ so that no algorithm can solve 3-SAT instances with $m$ clauses in time $O(2^{\delta{}m})$.
This modification of the commonly known version of ETH is in fact equivalent to that, see e.g. \cite[Theorem 14.4]{cygan2015}.

\begin{theorem}
\label{N-logsquare}
Assuming that ETH is true, then for every $\epsilon>0$ there is no polynomial algorithm for 
{\sc weak short 2-linkage} problem when the input $D$ is a digraph on $n$ vertices and $k_1,k_2=\Theta(\log^{1+\epsilon}{}n)$.
\end{theorem}

\pf We give the proof when  $k_1,k_2=\log^{1+\epsilon}{}N$, where $N$ is the number of vertices in the input digraph. Let $[D,s_1,s_2,t_1,t_2]$ be an instance of the weak 2-linkage problem and let $n$ be the number of vertices of $D$. 
Let $\epsilon'$ be defined such that $2^{n^{\frac{1}{1+\epsilon'}}} = \left\lceil 2^{n^{\frac{1}{1+\epsilon}}} \right\rceil$ and 
note that $\epsilon' \leq \epsilon $ and that $\epsilon' > 0$ when $n$ is large enough.
 Construct a new digraph $D'$ by adding an independent set of size
$2^{n^{\frac{1}{1+\epsilon'}}}-n$ so that the resulting digraph has $N=2^{n^{\frac{1}{1+\epsilon'}}}$ vertices, implying that we have 
$(\log{}N)^{1+\epsilon'}=n$. Clearly every pair of arc-disjoint $(s_1,t_1)$-, $(s_2,t_2)$-paths, $P_1,P_2$  in $D'$ use only vertices from $D$ and hence each of their lengths is at most $n = \log^{1+\epsilon'}{}N \leq \log^{1+\epsilon}{}N$ so $P_i$ is  at most
$k_i$ longer that the shortest $(s_i,t_i)$-path for $i=1,2$.

Suppose there is an algorithm for the {\sc weak short 2-linkage} problem that runs in time $O(N^c)$ for inputs on $N$ vertices when $k_1,k_2=\log^{1+\epsilon}{}N$ for some fixed constanct $c>0$. Then we have
\begin{eqnarray*}
  N^c  &=& (2^{n^{\frac{1}{1+\epsilon'}}})^c\\
       &=&  2^{c\cdot{}n^{\frac{1}{1+\epsilon'}}}\\
       &<&  2^{\delta\cdot{}n}\\    
\end{eqnarray*}

\noindent{}For every fixed constant $\delta{}>0$ provided that $n$ is large enough. This means that we can solve the general weak 2-linkage problem in time $O(2^{\delta\cdot{}n})$ for every $\delta>0$.

To see that this contradicts the ETH, we just have to observe that the reduction from 3-SAT to the 2-linkage problem in \cite{fortuneTCS10} (see also \cite[Section 10.2]{bang2009}) converts a 3-SAT formula with $n$ variables and $m$ clauses into an instance of 2-linkage with at most $dm$ vertices where $d$ is a constant (it is at most 61). Furthermore, the 2-linkage problem for a digraph on $n$ vertices reduces to the weak 2-linkage problem on a digraph with twice as many vertices. \qed

\section{Remarks and open problems}

Slivkins \cite{slivkinsSJDM24} proved that the {\sc weak $k$-linkage} problem is $W[1]$-hard for acyclic digraphs. We can prove that the same holds for {\sc short weak $k$-linkage} in acyclic digraphs. Indeed, consider an instance of {\sc weak $k$-linkage} on an acyclic digraph $D$ and consider a topological ordering $v_1,...,v_n$ of the vertices of $D$, \textit{i.e.} an ordering such that for every arc $v_iv_j$, we have $j>i$. Let us build $D'$ from $D$ by replacing every arc $v_iv_j$ in $D$ by a directed path of length $(j-i)$ in $D'$. Hence, $D'$ is still acyclic and every walk between a vertex $v_i$ and a vertex $v_j$ in $D$ is now replaced by a walk of length  $j-i$ in $D'$ and is thus a shortest walk. Therefore, a solution of short weak $k$-linkage in $D'$ immediately provides a solution of weak $k$-linkage in $D$.


\bibliography{refs}


\end{spacing}

\end{document}